\documentclass{aamas2013}
 \pdfoutput=1

\usepackage[T1]{fontenc}
\usepackage{ae,aecompl}
\usepackage{microtype}
\usepackage{amssymb}
\usepackage{subfig}
\usepackage{tikz}
\usepackage{syntax}
\usepackage{multirow}
\usepackage{booktabs}
\usepackage[scientific-notation=true]{siunitx}
\usepackage{url}

\usetikzlibrary{shapes,arrows}
\tikzstyle{bigworld} = [circle, draw, fill=white!20, text width=4.5em, text badly centered, node
distance=3cm, inner sep=0pt]
\tikzstyle{smallworld} = [circle, draw, fill=white!20, text width=2em, text badly centered, node
distance=1.5cm, inner sep=0pt]
\tikzstyle{line}=[draw,-latex']

\newcommand{\svs}{\: | \:}
\newtheorem{theorem}{Theorem}[section]
\newtheorem{lemma}[theorem]{Lemma}

\newtheorem{corollary}[theorem]{Corollary}
\newtheorem{definition}[theorem]{Definition}
\newtheorem{example}[theorem]{Example}

\title{Automatic Verification of Parameterised Interleaved Multi-Agent
  Systems}

\numberofauthors{2}
\author{
\alignauthor
Panagiotis Kouvaros 
\affaddr{\\Department of Computing\\Imperial College London}
    \email{panagiotis.kouvaros10@imperial.ac.uk}
\alignauthor
Alessio Lomuscio
    \affaddr{\\Department of Computing\\Imperial College London}
    \email{a.lomuscio@imperial.ac.uk}
}
\begin{document}

\maketitle

\begin{abstract}
A key problem in verification of multi-agent systems by
model checking concerns the fact that the state-space of the
system grows exponentially with the number of agents
present. This often makes practical model checking
unfeasible whenever the system contains more than a few
agents. In this paper we put forward a technique to
establish a cutoff result, thereby showing that systems with
an arbitrary number of agents can be verified by checking a
single system consisting of a number of agents equal to the
cutoff of the system. While this problem is undecidable in
general, we here define a class of parameterised interpreted
systems and a parameterised temporal-epistemic logic for
which the result can be shown. We exemplify the theoretical
results on a robotic example and present an implementation
of the technique as an extension of \texttt{\textsc{mcmas}},
an open-source model checker for multi-agent systems.
\end{abstract}

\category{D.2.4}{Software/Program Verification}{Model Checking}
\terms{Theory; Verification}
\keywords{Epistemic Logic; Model Checking; Parameterised MAS}

\section{introduction}
Verification and validation of systems before deployment is
increasingly seen of fundamental importance not just in
safety-critical applications, but also in more mainstream
applications. Multi-Agent Systems (MAS) are no exception. The past ten
years have witnessed considerable research in verification techniques
aimed at assessing automatically whether or not a MAS meets its
intended specifications.

One of the leading techniques in this area is \emph{model
checking}~\cite{clarke+99a}. In this setting the system $S$
under analysis is encoded as a transition system $M_S$ and a
specification $P$ is formalised as a logical formula
$\phi_P$; a model checking procedure is used to determine
whether $M_S \models \phi_P$, i.e., whether or not the
system $M_S$ satisfies the formula $\phi_P$. Since MAS
specifications are often expressed as epistemic, deontic,
and ATL formulas; the techniques put forward in the MAS
community reflect
this~\cite{kacprzak2008verics,gammie2004mck,lomuscio2009mcmas}.
While explicit techniques are less efficient, symbolic
checkers such as \texttt{\textsc{MCK}}~\cite{gammie2004mck},
\texttt{\textsc{mcmas}}~\cite{lomuscio2009mcmas} and
\texttt{\textsc{VerICS}}~\cite{kacprzak2008verics} are
capable of handling state-spaces of the region of $10^{15}$
and beyond. However, MAS-based applications, due to the
agents' complex and intentional nature, often generate much
larger state spaces. To alleviate this problem, a number of
techniques, including
abstraction~\cite{cohen2009abstraction}, have been put
forward. These have been successful in allowing users to
tackle larger systems, but it is still the case that,
generally speaking, systems with many agents are difficult
to verify. The aim
of this paper is to contribute to overcome this shortcoming.

In the present setting we consider a class of MAS composed
of identical agents interacting with the environment.  While
this may seem a strong condition, this is a relatively
common assumption in many application areas of interest
ranging from robotics, to artificial life, swarm
intelligence, services and in open systems in general. A
natural question in these systems is whether certain
properties hold \emph{irrespective of the number of agents
present}. For example, in a remote robotic scenario we may
wish to check that a goal is met irrespective of how many
robots are present. It is immediate to see that plain
 model checking cannot be used to solve this problem. To establish
this property we would have to consider an infinite number
of different systems each composed of a different number of
agents and run model checking algorithms on each of these.
Not withstanding the fact that we cannot check an infinite
number of systems, this class includes instances
for which model checking would
require an unfeasible amount of memory and time.

In this paper we develop a technique that enables us to
derive the number of agents that is sufficient to consider
to show that a property holds in the system \emph{for any
number of agents}. In line with the literature on reactive
systems~\cite{emerson+00,hanna+10}, we call this bound the
\emph{MAS cutoff}. In contrast with literature in reactive
systems we here work with \emph{interleaved interpreted
systems}~\cite{lomuscio+10a} and  temporal-epistemic
specifications.

The rest of the paper is organised as follows. In Section~2
we define an interleaved semantics that will use throughout
the paper, a logic that we call IACTL$^*$K$_{-X}$ which
combines the universal fragment of CTL$^*$ without ``next''
with a parameterised version of epistemic logic, and
establish a stuttering-equivalence simulation result.
Section~3 introduces the technique to establish the cutoff
and presents our main theoretical result. To exemplify the
theory we discuss a robotic example in Section~4. We discuss
our implementation in Section~5 and present experimental
results. We conclude in Section~6 also discussing related
work.

\section{parameterised interleaved \\ interpreted systems}

In this section we introduce a framework for reasoning about
parameterised multi-agent systems. In particular, we recall the
semantics of interleaved interpreted systems~\cite{lomuscio+10a} and
we introduce parameterised interleaved multi-agent systems. To
reason about the temporal-epistemic properties of agents, we introduce
the logic IACTL\(^*\)K\(_{-X}\), a parameterised extension of
ACTL\(^*_{-X}\) with indexed atomic propositions and indexed epistemic
modalities.

\subsection{Interleaved Interpreted Systems}

The interpreted systems (IS) formalism~\cite{fagin+03a} is a standard
semantics for MAS. Here we consider a special class of interpreted
systems, called interleaved interpreted systems
(IIS)~\cite{lomuscio+10a}, in which the agents evolve in parallel
asynchronously (i.e., by means of interleaving
semantics~\cite{clarke+99a}). Differently from standard interpreted
systems where actions may be performed by all the agents at the same
round, IIS insist on only one local action at the time to be performed
in the system. If at any given round more than one agent admits in
its repertoire the action to be performed, then all agents sharing
this action  perform this action at that round. Thus, the agents
communicate by means of shared actions. The temporal evolution of an
agent's local states is accommodated to the needs of interleaving;
while in standard IS the next local state depends on the actions
performed by all agents in the system, in IIS local states depend only
on the agent's own action.  Below, we summarise the framework of IIS, as
presented in~\cite{lomuscio+10a}, to model interleaved MAS.

We assume that a MAS is composed of \( n \) agents \( \mathcal{A} =
\{1,\ldots\,n\} \). Each agent \( i \in \mathcal{A} \) is
characterised by a finite set of local states \( L_{i} \) and a finite
set of actions \( Act_{i} \). Each \( Act_{i} \) contains a special
action \( \epsilon_{i} \) which we call the ``silent'' action; as the
name suggests, whenever \(\epsilon_{i}\) is performed, agent \(i\)'s
local state does not change.  We call \(ACT =\bigcup_{i \in
    \mathcal{A}} Act_{i}\) the union of all actions. Actions are
    performed in compliance with a protocol \( P_{i} : L_{i}
    \rightarrow \wp(Act_{i}) \) governing which actions can be
    executed in a given state. The silent action is enabled at every
    local state; formally, \( \forall i \in \mathcal{A} : \forall
    l_{i} \in L_{i} : \epsilon_{i} \in P_{i}(l_{i}) \). For each
    action \(a\), we call \(Agent(a) = \{i \in \mathcal{A} \svs a \in
    Act_{i} \}\) the set of agents potentially able to perform \(a\).
    The evolution of agent \(i\)'s local states is described by
    the transition function \( t_{i} : L_{i} \times Act_{i}
    \rightarrow L_{i} \) such that \( t_{i}(l_{i},\epsilon_{i}) =
    l_{i} \) for each \( l_{i} \in L_{i} \). Note that \(t_{i}\) is a
    function of agent \(i\)'s local action only.

A global state \( g=(l_{1},\cdots,l_{n}) \in L_{1} \times \cdots
\times L_{n} \) is an \(n\)-tuple of local states for all the agents
in the MAS and represents the state of the system at a particular
instance of time. Given a global state \(g=(l_{1},\ldots,l_{n})\), we
write \(g_{i}\) to denote the local component \(l_{i}\) of agent \(i
\in \mathcal{A}\) in \(g\). Given a set of agents \(J = \{
j_{1},\dots,j_{|J|}\} \subseteq \mathcal{A}\), we write \(g_{J}\) to
denote the tuple of local components \( (l_{j_{1}},\ldots,l_{j_{|J|}})
\) of agents \( J \) in \(g\).
The local protocols and the local evolution functions determine how
the system proceeds from one global state to the next.

\begin{definition}{(Interleaved Semantics)}
	\label{def:isemantics}
    Let G be a set of global states. The global interleaved evolution
    function \(t: G \times Act_{i} \times \cdots \times Act_{n}
    \rightarrow G\) is defined as follows:
    \(t(g,a_{1},\ldots,a_{n})=g'\) iff there exists an action \(a
    \in ACT\) such that for all \(i \in Agent(a)\) we have that
    \(a_{i}=a\) and \(t_{i}(g_{i},a)=g'_{i}\); and for all \(i \in
    \mathcal{A} \setminus Agent(a)\), we have that \(a_{i} =
    \epsilon_{i}\) and \(t_{i}(g_{i},a_{i}) = g'_{i} = g_{i} \).  In
    short, we write the above as \(g \stackrel{a}{\rightarrow} g'\).
\end{definition}

We assume that the joint silent action is always enabled; thus, the
global transition relation is serial. A sequence of global states and
actions \(\pi = g^{1}a^{1}g^{2}a^{2}g^{3} \ldots\) is said to be an
interleaved path (or simply a path) originating at \(g^{1}\) if for
every pair of successor states we have that \(g^{i}
\stackrel{a^{i}}{\rightarrow}g^{i+1}\), for every \(i\geq 1\).  We
write \(\pi(i)\) to denote the \(i\)-th global state in \(\pi\). The
set of all  paths originating from \(g\) is denoted by
\(\Pi(g)\). The local path of agent \(i \in \mathcal{A} \)
in \(\pi\) is the projection of \(\pi\) onto \(i\); i.e., the sequence
\( \pi^{i} = g^{1}_{i}a^{1}g^{2}_{i}a^{2}g^{3}_{i} \ldots\).  The
projection of \(\pi\) onto a set of agents \(J\) is the sequence
\(\pi^{J} = g^{1}_{J}a^{1}g^{2}_{J}a^{2}g^{3}_{J} \ldots\). We denote
by \(\pi[i]\) the suffix \(g^ia^ig^{i+1}\cdots\) of \(\pi\). A state
\(g \in G\) is said to be reachable from \(g^{1} \in G\) if there is a
path \(g^{1}a^{1}g^{2}\cdots\) such that \(g=g^{i}\), for some \(i
\geq 1\).

\begin{definition}{(Interleaved Interpreted Systems)}
    \label{def:IIS}    
    Let \(\mathit{AP}\) be a set of atomic propositions.  An \emph{interleaved
      interpreted system} (IIS), or a model, is a \(4\)-tuple
    \(\mathcal{M}=\langle G,\iota,\Pi,V \rangle\), where \(G\) is a
    set of global states, \(\iota \in G\) is an initial global state
    such that each state in \(G\) is reachable from \(\iota\), \( \Pi
    = \displaystyle \bigcup_{g \in G} \Pi(g) \) is the set of all
    interleaved paths originating from all states in \(G\), and \(V:
    \mathit{AP} \rightarrow \wp(G) \) is a valuation function.
\end{definition}

\begin{example} \label{ex:me}

The IIS presented in Figure~\ref{fig:ex:me} is a modified version of
the autonomous robot (AR) example from~\cite{fagin+03a}. A robot runs
along an endless straight track; its position is given in terms of
locations numbered as \(0,1,2,\cdots\). The robot can only move
forward along the track starting at position 0. A faulty sensor is
attached to the robot measuring its position; a sensor reading at
location \(q\) can be any of the values in \(R_{q} = \{q-1,q,q+1\}\).
The movement of the robot is controlled by the environment. The only
action the robot can perform is to halt; if the robot does not halt,
the environment may move the robot one position forward at each time
step; once the robot halts, the environment can no longer move it.
The goal of the robot is never to exit the goal region \( \mathit{GR}
= \{2,3,4\}\) upon entering into it and never to halt in the
restricted region \(\mathit{RR}=\{0,1\}\) . A sound and complete
solution to the autonomous robot problem~\cite{fagin+03a} is for the
robot to do nothing while the value of its sensor is less than 3 and
to halt once the value of its sensor is greater than or equal to 3. In
the figure a state of the environment represents the position of the
robot, and a state $psh$ of the robot represents, respectively, its
position, its sensor reading, and whether it has halted or not.

\begin{figure}
	\centering
    \scalebox{0.95}{
    \subfloat[Environment]{
		\begin{tikzpicture}[node distance=2cm, auto]
			\node [smallworld] (zero) {$0$};
			\node [smallworld, right of=zero] (one) {$1$};
            \path [line,font=\scriptsize] (zero) edge node {\(m^{=}\)} (one);
            \path [line,font=\scriptsize] (zero) edge [bend right=60] node[below]
            {\(m^{-}\)} (one);
            \path [line,font=\scriptsize] (zero) edge [bend
            left=60] node {\(m^{+}\)} (one);
            \draw[dotted] (2,0) -- (2.5,0);


        \end{tikzpicture}
			\label{fig:process1}
		}
		
		\subfloat[Robot]{
			\begin{tikzpicture}[node distance=2cm, auto]
                \node [smallworld] (zz) {$00\bot$};
				\node [smallworld, right of=zz, above of=zz] (ot)
                {$12\bot$};
            	\node [smallworld, right of=zz] (oo)
                {$11\bot$};
            	\node [smallworld, right of=zz, below of=zz] (oz)
                {$10\bot$};
				\node [smallworld, right of=oo] (tt)
                {$23\bot$};
                \node [smallworld, right of=tt] (h)
                {$23\top$};

                \path [line,font=\scriptsize] (zz) edge
                node  {\(m^{+}\)} (ot);
               \path [line,font=\scriptsize] (zz) edge
                node  {\(m^{=}\)} (oo);
               \path [line,font=\scriptsize] (zz) edge
               node[below left]  {\(m^{-}\)} (oz);

                \path [line,font=\scriptsize] (ot) edge
                node  {\(m^{+}\)} (tt);
                \path [line,font=\scriptsize] (oo) edge
                node  {\(m^{+}\)} (tt);

              \path [line,font=\scriptsize] (oz) edge
               node[below right]  {\(m^{+}\)} (tt);
             \path[line, font=\scriptsize] (tt) edge[bend left=60]
           node {$halt$} (h);
             \draw[dotted] (2,1.5) -- (2.5,1.5); 
            \draw[dotted] (2,1.5) -- (2.5,1.5); 

             \draw[dotted] (3.5,0) -- (3.7,0);
            \draw[dotted] (2,-1.5) -- (2.5,-1.5);

			\end{tikzpicture}
			\label{fig:controller}
		}}
        \caption{The IIS for the Autonomous Robot Example.}
		\label{fig:ex:me}
        \vspace{-0.5cm}
	\end{figure}
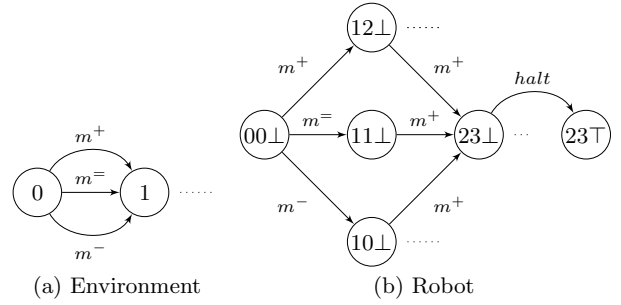

\end{example}

\subsection{Template Agent and Parameterised Interleaved Systems}

Several protocols are designed for an unbounded number of identical
participants. Cache coherence, mutual exclusion, and voting protocols
are typical examples in which the number of participants
(caches, processes, and voters respectively) is independent of the
design process. Multi-party
negotiation protocols, auctions and open MAS in general also have this property. In the
following we develop a formal model of a parameterised multi-agent
system, composed of an arbitrary number of identical agents, that can
be used in these circumstances. Given that the number of agents is a
priori unknown, a parameterised system describes an infinite family of
systems where an instance in the family, or concrete instantiation, is
obtained by specifying the number of agents in the system. Formally,
we introduce below parameterised interleaved interpreted systems
(PIIS), an extension of interleaved interpreted systems, to model the
aforementioned classes of systems.

We write \(\mathcal{T}(n)\) to denote a PIIS, where \(n \geq 1 \) is
the parameter specifying the number of agents, each constructed from a
template agent \(\mathcal{T}\).  The template agent is an interleaved
agent encoded with a set of synchronous actions and a set of
asynchronous actions. As it will be clear below, if the action
performed in a global transition is a synchronous action, then all
agents participate in the global action by performing the same
synchronous action.  However, if the action performed in a global
transition is an asynchronous action, then exactly one agent
participates in the global action.  Therefore, all agents synchronise
at any time step in which a synchronous action is performed.

\begin{definition}{(Template Agent)} \label{def:template-agent}
    Given a set of propositions \(\mathit{AP}\), a \emph{template
    agent} is a tuple \(\mathcal{T} = \langle L,\iota,Act,P,t,h
    \rangle\), where \(L\) is a finite nonempty set of template states
    from which  \( \iota \in L \) is the unique initial template
    state,  \( Act = Act^{S} \cup Act^{A} \cup \epsilon \) is a finite
    set of template actions, where \(Act^{S}\) is a set of synchronous
    actions, \(Act^{A}\) is a set of asynchronous actions with 
    \(Act^{S} \cap Act^{A} \cap \{\epsilon\} = \emptyset\), 	\( P:
    L \rightarrow \wp(Act) \) is the protocol such that for all \( l
    \in L \), \(\epsilon \in P(l) \), \( t : L \times Act \rightarrow
    L \)  is the deterministic template evolution function such that
    for all \(l \in L\), \(t(l,\epsilon)=l \), and  \( h : L \rightarrow
    \wp(\mathit{AP}) \) is a labelling function for the template
    states.
\end{definition}

Given a template agent \(\mathcal{T}\), \(\mathcal{T}(n)\) denotes the
parallel composition of \(n\) concrete agents\footnote{When it is
  clear from the context, we write ``agent'' instead of ``concrete
  agent''.}  in \(\mathcal{A}=\{1,\cdots,n\}\). Each agent \(i \in
\mathcal{A}\) is obtained by subscripting the states and actions of
\(\mathcal{T}\) as follows: \( L_{i} = L \times \{i\} \), \( Act_{i} =
Act^{S} \cup Act_{i}^{A} \cup \epsilon_{i}\), where \(Act_{i}^{A} =
Act^{A} \times \{i\} \); synchronous template actions are not
subscripted. For a concrete action \(a \in Act_{i}\), we write
\(tl(a)\) to refer to the corresponding template action; analogously,
for a concrete state \(l_{i} \in L_{i}\), we write \(tl(l_{i})\) to
refer to the corresponding template state \(l\). The local protocol \(
P_{i}: L_{i} \rightarrow \wp(Act_{i}) \) of the \(i\)-th agent is
defined by \( a \in P_{i}(l_{i}) \) iff \( tl(a) \in P(l) \). The
evolution function \( t_{i}: L_{i} \times Act_{i} \rightarrow L_{i} \)
of the \(i\)-th agent is defined by \( t_{i}(l_{i},a)=l_{i}' \) iff \(
t(l,tl(a))=l' \). We associate with each \(i \in \mathcal{A}\) a local
labelling function \(V_{i} : L_{i} \rightarrow \wp(\mathit{AP} \times \{i\})\)
defined by \(p_{i} \in V_{i}(l_{i})\) iff \(p \in h(l)\).

The global transitions we consider in PIIS are as in
Definition~\ref{def:isemantics}.  A global transition from a global
state \(g\) complies with the definition if either a synchronous
action is enabled for all agents in \(g\) or an asynchronous action
\(a_{i} \in Act_{i}\) is enabled for an agent \(i \in \mathcal{A}\) in
\(g\).  Indeed, if \(a \in Act^{S}\), then \(Agent(a) = \mathcal{A}
\); if \(a_{i} \in Act_{i}^{A} \), for some \(i \in \mathcal{A}\),
then \(Agent(a) = \{i\}\). Therefore, synchronous actions play the
role of shared actions in IIS, but here synchronous actions are
shared by all agents.  We now define parameterised interleaved
interpreted systems.

\begin{definition}{(Parameterised Interleaved Interpreted Systems)} 
    \label{def:parameterised-system} 
    Given a natural number \(n \geq 1\) and a template agent
    \(\mathcal{T} = \langle L,\iota,Act,P,t,h \rangle\), a
    \emph{parameterised interleaved interpreted system} (PIIS),
    composed of \(n\) concrete agents, is a tuple \( \mathcal{T}(n) =
    \langle G^{n},\iota^{n},\Pi^{n},V^{n} \rangle\), where \(G^{n} = L
    \times [n]\) is a set of global states, \(\iota^{n} =
    (\iota_{1},\ldots,\iota_{n})\) is an initial (global) state,
    \(\Pi^{n} = \displaystyle \bigcup_{g \in G^{n}} \Pi(g)\) is
    the set of all interleaved paths originating from all states in
    \(G^{n}\), and \(V^{n}: G^{n} \rightarrow \wp(\mathit{AP} \times
    \mathcal{A})\) is a labelling function defined by \(p_{i} \in
    V^{n}(g)\) iff \(p_{i} \in V_{i}(g(i))\).
\end{definition}

Given a template agent, the above definition denotes an infinite family
of concrete systems. A member of the family, which we call an instance
of the parameterised system, is obtained by fixing the value of the
parameter \(n\).

\subsection{The Specification Language IACTL\(^*\)K\(_{\mathbf{-X}}\)} 

Temporal-epistemic logic has been widely adopted to express the
properties of agents in a MAS. 
However, we cannot use propositional temporal-epistemic logics
to reason about an unbounded number of agents. To see this, consider
the parameterised variant of the autonomous robot and suppose that
we want to express the property: ``for every \(i \in \mathcal{A}\),
whenever \(i\) halts, then it knows that every other robot's position 
is within the goal region''. This property encodes all distinct
pairs of robots.  Therefore, to express the property for an AR
composed of \(n\) robots we need to construct a  formula
composed of \(2!\binom{n}{2}\) conjuncts.  Instead we would like to
express properties that are independent of the number of agents in the
system, as if we were able to quantify over the agents. To overcome these
shortcomings we introduce the indexed temporal-epistemic logic
IACTL\(^*\)K\(_{-X}\).  Indexed logics are commonly used in
parameterised systems~\cite{emerson+00}.

IACTL\(^*\)K\(_{-X}\) combines indexed epistemic modalities with the
universal fragment of CTL\(^*_{-X}\)~\cite{clarke+89a} (the logic
CTL\(^*\), without the next-time operator, extended with indexed
atomic propositions).  We consider a stuttering-insensitive logic,
i.e., a logic insensitive to repeated occurrences of the same state,
or equivalently a logic without the next-time operator
\cite{clarke+99a}). This is because the next-time operator can be used to count
the number of agents in the system~\cite{clarke+89a,emerson2003model}
resulting in the parameterised verification problem being
undecidable~\cite{emerson2003model}.

Intuitively, any IACTL\(^*\)K\(_{-X}\) formula \(\varphi\) represents
an ACTL\(^*\)K\(_{-X}\) formula for each concrete system
\(\mathcal{T}(n)\), \(n \geq c\), where \(c\) is the number of unique
indices contained in \(\varphi\); the formula corresponding to
\(\mathcal{T}(n)\) is the conjunction of all formulae that can be
constructed from \(\varphi\) by instantiating the indices with every
\(c\)-tuple of distinct agents in \(\mathcal{T}(n)\).

\subsubsection{Syntax and semantics of IACTL\(^*\)K\(_{-X}\)}

We assume a set \(\mathit{VS}\) of variable symbols which we use to index the
atomic propositions and the epistemic modalities.  There are two types
of formulas in IACTL\(^*\)K\(_{-X}\): (i) state formulas which are true
at a state and (ii) path formulas which are true on a path.

\begin{definition}{(Syntax of IACTL\(^*\)K\(_{-X}\))}
    \label{def:iactl*-syntax}
    The state and path formulae of IACTL\(^*\)K\(_{-X}\) over a set
    \(\mathit{AP}\) of propositions and a set \(\mathit{VS}\) of variable symbols are
    inductively defined as follows:
	\begin{itemize}
        \item S1. if \( p \in \mathit{AP} \) and \(v \in \mathit{VS}\), then  \(p_{v}\)
            and \(\neg p_{v}\) are state formulas;
        \item S2. if \( \varphi \) and \( \psi\) are state formulas,
            then  \( \varphi \wedge \psi \), \( \varphi \vee \psi \)
            and \(K_{v}\varphi\) (\(v \in \mathit{VS} \)) are  state formulas;
        \item S3. if \(\varphi\) is a state formula with
            exactly \(J \subseteq \mathit{VS}\) variable symbols, then
            \(\bigwedge_{J} \varphi\) is a state formula;
        \item S4. if \(\varphi\) is a path formula, then
            \(A(\varphi)\) is a state formula;
		\item P1. any state formula \(\varphi\) is also a path formula;
        \item P2. if \(\varphi\) and \(\psi\) are path formulas, then
            \(\varphi \wedge \psi\) and \(\varphi \vee \psi\) are
            path formulas;
        \item P3. if \(\varphi\) and \(\psi\) are path formulas, then
            \(U(\varphi,\psi)\) and \(R(\varphi,\psi)\) are  path
            formulas.
	\end{itemize}
\end{definition}

The \(\bigwedge_{J}\) connective serves as a universal agent
quantifier ranging over all \(|J|\)-tuples of pairwise distinct
agents. Given a formula \(\varphi\), a variable \(v \in \mathit{VS}\),
occurring in \(\varphi\), is said to be bound if it is in the scope of
a \(\bigwedge_{J}\) connective; otherwise, \(v\) is said to be free. A
formula in which there are no free occurrences of variables is said to
be a sentence. We here consider only sentences. For an
IACTL\(^*\)K\(_{-X}\) formula \(\varphi\), we write \(\varphi(J)\) to
indicate that: (i) all variables in \(J \subseteq \mathit{VS}\) and
only them occur free in \(\varphi\), and (ii) \(\varphi\) does not
contain any \(\bigwedge_{J}\) connectives.  The path quantifier \(A\)
stands for ``for all paths''. The temporal operators \(U\) and \(R\)
represent ``until'' and ``release'' respectively; the formula
\(U(\varphi,\psi)\) is read as ``\(\varphi\) holds continuously until
\(\psi\) holds'', whereas the formula \(R(\varphi,\psi)\) is read as
``\(\varphi\) releases \(\psi\)''. The operator \(K\) denotes the
epistemic modality; \(K_{v}\varphi\) is read as ``each concrete agent
\(i \in \mathcal{A}\) knows \(\varphi\)''. Since we consider sentences
only, \(v\) is always bound by a \(\bigwedge_{J}\) connective;
therefore \(v\) ranges over all agents.

Consider the AR again; we can now easily express the properties
previously stated by considering the IACTL\(^*\)K\(_{-X}\) formula
\( \varphi_{AR1} = \bigwedge_{\{i,j\}} AG \left(h_i \rightarrow K_i
 g_j \right)\), where the atomic propositions $h_i, g_i$ hold
 respectively in the states where robot $i$ has halted and is within
 the goal region.

The specifications we consider in this paper are of the form
\(\bigwedge_{J} \varphi(J)\).  Since these formulas range over all
\(|J|\)-tuples of distinct agents, a standard model checking procedure
would have to consider every instantiation of \(\varphi(J)\). However,
a result we obtain is that model checking a formula \(\bigwedge_{J}
\varphi(J) \in\)  IACTL\(^*\)K\(_{-X}\) \emph{can be reduced to model
checking a single instantiation of} \(\varphi(J)\), thereby
simplifying the complexity of the model checking procedure. Note that
an instantiation of \(\varphi(J)\) is an ACTL\(^*\)K\(_{-X}\) formula,
built as follows:

\begin{definition}
    ACTL\(^*\)K\(_{-X}\) formulae over a set \(\mathit{AP}\) of atomic
    propositions and a set \(\mathcal{A}\) of agents are defined as in
    Definition~\ref{def:iactl*-syntax} but omitting (S3) and replacing
    (S1) and (S2) with:
         S1'. if \( p \in \mathit{AP} \) and \(i \in \mathcal{A}\), then
            \(p_{i}\) and \(\neg p_{i}\) are state formulas;
         S2'. if \( \varphi \) and \( \psi\) are state formulas,
            then  \( \varphi \wedge \psi \), \( \varphi \vee \psi \)
            and \(K_{i}\varphi\) (\(i \in \mathcal{A} \)) are  state
            formulas;
\end{definition}

For an ACTL\(^*\)K\(_{-X}\) formula \(\varphi\), we write
\(\varphi(J)\) (\(J \subseteq \mathcal{A}\)) to indicate that for each
subformula \(K_{i}\psi\) and each proposition \(p_{j}\) of \(\varphi\)
we have that \(i,j \in J\). We write ACTL\(^*\)K\(_{-X}^{J}\) for the
restriction of ACTL\(^*\)K\(_{-X}\) to all formulae of the form
\(\varphi(J)\).
We interpret IACTL\(^*\)K\(_{-X}\) formulae over PIIS\@. The temporal
modalities are interpreted over the global transition relation and the
epistemic modalities are interpreted over the equality of the local
components of the global states.

\begin{definition}{(Satisfaction)}
 \label{def:ctl*-semantics}
Let \(\mathcal{T}(n) = \langle G,\iota,\Pi,V \rangle\) be a
parameterised interleaved interpreted system, let
\(\pi=g^{1},a^{1},g^{2},\ldots\) be a path of \(\mathcal{T}(n)\), let
\(g \in G\) be a state of \(\mathcal{T}(n)\), and let \(\varphi\) be
an IACTL\(^*\)K\(_{-X}\) formula.  Satisfaction of \(\varphi\) at g,
denoted \( (\mathcal{T}(n),g) \models \varphi\), or simply \(g \models
\varphi\), and satisfaction of \(\varphi\) on \(\pi\), denoted
\((\mathcal{T}(n),\pi) \models \varphi\), or just \(\pi \models
\varphi\), is inductively defined as follows:

\begin{table}[h]
    \begin{tabular}{lllp{4.4cm}}
        S1. & \( g  \models p_{i}\) &  iff & \( p_{i} \in V_{i}(g_{i})
        \); \\
        & \( g  \models \neg p_{i}\) &  iff & \( \text{not } g \models
        p_{i} \), for \(p_{i} \in \mathit{AP} \times \mathcal{A}\); \\
        S2. & \( g \models \varphi \wedge \psi\) & iff & \( g \models
        \varphi \text{ and } g \models \psi\); \\
         & \( g \models \varphi \vee \psi\) & iff & \( g \models
         \varphi \text{ or } g \models \psi\); \\
         & \( g \models K_{i} \varphi \) & iff & \(g' \models
         \varphi\) for every \(g' \in G\) such that \(g_i = g'_i\);
         \\
        S3. & \(g \models \bigwedge_{J} \varphi(J)\) & iff & \(g
        \models \varphi(C)\) for every \(C \in \{I \; | \; I \subseteq
            \mathcal{A} \text{ and } |I|=|J|\}\); \\
        S4. & \(g \models A \varphi\) & iff & \(\pi \models \varphi\)
        for every path \(\pi\) such that \(\pi(1)=g\); \\
        P1. & \( \pi \models \varphi \) &  iff & \( \pi(1) \models
        \varphi \) for any state formula \(\varphi\); \\
        P2 & \( \pi \models \varphi \wedge \psi\) & iff & \(\pi
        \models \varphi \text{ and } \pi \models \psi\); \\
        & \( \pi \models \varphi \vee \psi\) & iff & \(\pi
        \models \varphi \text{ or } \pi \models \psi\); \\

        P3 & \(\pi \models U(\varphi,\psi)\) & iff &  there is  an \(i
        \geq 1\) such that \( \pi[i] \models \psi \)  and \(\pi[j]
        \models \varphi\) for all \(1 \leq j < i\); \\
        & \( \pi \models R(\varphi,\psi)\) & iff & for every \(i\), if
        \( \pi[j] \nvDash \varphi \), for all \(1 \leq j < i \), then
        \(\pi[i] \models \psi\).
	\end{tabular}
\end{table}
\end{definition}
\vspace{-0.5cm}
We use the following abbreviations: \( \top \stackrel{def}{\equiv}
p_{v} \vee \neg p_{v} \), \(\bot \stackrel{def}{\equiv} p_{v} \wedge
\neg p_{v}\), for some \(p \in \mathit{AP}\) and \(v \in
\mathit{VS}\), \(F\varphi \stackrel{def}{\equiv} U(\top,\varphi)\)
(``Eventually \(\varphi\)''), \(G\varphi \stackrel{def}{\equiv}
R(\bot,\varphi)\) (``Always \(\varphi\)''). A formula \(\varphi\) is
said to be true in \(\mathcal{T}(n)\), denoted \(\mathcal{T}(n)
\models \varphi\), iff \((\mathcal{T}(n),\iota) \models \varphi\).

\subsubsection{Symmetry reduction for ACTL\(^*\)K\(_{-X}\)}

Symmetry reduction techniques have been used to reduce the complexity
of model checking temporal-epistemic properties of multi-agent
systems~\cite{cohen2009symmetry}. Since a PIIS is
composed of identical agents, intuitively, there is an inherent
symmetry in the system that we can exploit. Indeed, we adapt a 
result from reactive systems~\cite{emerson+96b} and we show that
an IACTL\(^*\)K\(_{-X}\) formula \(\bigwedge_J \varphi(J)\) is
equivalent to a single instantiation \(\varphi(\{1,\ldots,|J|\})\) of
\(\varphi(J)\).


\begin{lemma} \label{lemma:symmetry-reduction}
\( \mathcal{T}(n) \models \bigwedge_{J} \varphi(J) \text{ iff }
\mathcal{T}(n) \models \varphi(\{1,\ldots,|J|\}).
\)
\end{lemma}
\begin{proof}{\emph{(Sketch)}} 
    (\(\Rightarrow\)) Obvious. (\(\Leftarrow\)) Suppose that
    \(\mathcal{T}(n) \models \varphi(\{1,\cdots,k\})\) and let
    \(J'=\{j_1,\cdots,j_k\}\) be an arbitrary set of \(k\) agents. Let
    \(\zeta: \mathcal{A} \rightarrow \mathcal{A}\) be a bijective
    mapping such that \(\forall i \in
    \{1,\cdots,k\}:\pi(i)=j_i\). Given an object \(o\) (either a state
    or a formula), let \(\zeta(o)\) denote the object \(o'\) obtained by
    replacing each occurrence of any \(i \in \mathcal{A}\) with
    \(\zeta(i)\). As \(g \stackrel{a}{\rightarrow} g'\) is iff \(\zeta(g)
    \stackrel{\zeta(a)}{\rightarrow}\zeta(g')\),   and \(g_i = g'_i\) is
    iff \((\zeta(g))_{\zeta(i)} = (\zeta(g'))_{\zeta(i)}\), we get that
    \((\mathcal{T}(n),\zeta(\iota)) \models \zeta(\varphi(\{1,\cdots,k\})\),
    and therefore \((\mathcal{T}(n),\zeta(\iota)) \models \varphi(J')\).
    As \(\zeta(\iota)=\iota\), we get that \((\mathcal{T}(n),\iota) \models
    \varphi(J)\), therefore \(\mathcal{T}(n) \models \bigwedge_J
    \varphi(J)\).
\end{proof}

Given the semantics of IACTL\(^*\)K\(_{-X}\), model checking a formula
\( \bigwedge_{J} \varphi(J)\) on a system $\mathcal{T}(n)$ is equivalent 
to model checking an ACTL\(^*\)K\(_{-X}\)
formula of \( c!\binom{c}{|J|} \) conjuncts of the form
\(\varphi(C)\), where \(C \in \{ I \; | \; I \subseteq \{1,\cdots,c\}
\text{ and } |I|=|J|\} \).  By the above lemma we can model check 
a single conjunct. For example, model
checking the formula \(\bigwedge_{\{i,j\}} AG \left(
h_i \rightarrow K_i g_j
\right)\) can be reduced to model checking the simpler 
formula \(AG ( h_{1} \rightarrow K_{1} 
g_{2})\).

\subsubsection{Invariance of ACTL\(^*\)K\(_{-X}\)}

As noted above, cutoff techniques concern the identification  of
a system instance, called the cutoff instance, which can be used to
check whether a given property holds for all system instances. A
notion of equivalence between the system instances is often used to
show this result.
Stuttering-insensitive logics are accompanied with the standard
notion of stuttering simulation~\cite{penczek+00}. A system stuttering
simulates another system if for every behaviour of the latter, there
is a stuttering equivalent behaviour of the former. Informally, two
behaviours are stuttering equivalent if the behaviours coincide when
each sequence of stutter steps (i.e., steps that do not affect the
labelling of the states) are collapsed onto a single step. Since we consider
only universal path quantification, it follows that any
ACTL\(^*\)\(_{-X}\) formula satisfied by the simulating model is also
satisfied by the simulated model. Below we extend the notion of
stuttering equivalence to ACTL\(^*\)K\(_{-X}\). Since our
specifications are of the form \(\varphi(J)\), referring only to
agents in \(J \subseteq \mathcal{A}\), we project the valuation function
onto \(J\). This projection, denoted \(V|_{J}\), is defined by
\(V|_{J}(g) = V(g) \cap \{p_{i} \svs i \in J\}\), for every state \(g
\in G\). 


\begin{definition}{(\(J\)-Stuttering simulation)} 
    \label{def:j-stuttering} 
    A relation \\ \(\sim_{Jss} \subseteq G \times G'\) is a
    \(J\)-stuttering simulation between two models
    \(\mathcal{M}=\langle G,\iota,\Pi,V \rangle\) and \(\mathcal{M}' =
    \langle G',\iota',\Pi',V' \rangle\) if the following conditions
    hold:
	\begin{enumerate}
		\item \(\iota \sim_{Jss} \iota'\);
		\item if \(g^1 \sim_{Jss}  g'^1\) then 
			\begin{enumerate}
                \item if \(g_i^1 = g^{2}_i\), for \(i \in J\), then
                    \(g'^1_i = g'^{2}_i\) for some \(g'^{2}\) such
                        that \(g^{2} \sim_{Jss} g'^{2}\);
            \item \(V|_{J}(g^1)=V|_{J}(g'^1)\) and for every path \(\pi\)
                of \(\mathcal{M}\) that starts at \(g^1\), there is a
                path \(\pi'\) of \(\mathcal{M}'\) that starts at
                \(g'^1\), a partition \(B_{1},B_{2} \ldots\) of \(\pi\),
                and a partition \(B_{1}',B_{2}',\ldots\) of \(\pi'\)
                such that for each \(j \geq 1\), \(B_{j}\) and
                \(B_{j}'\) are nonempty and finite, and every state in
                \(B_{j}\) is related by \(\sim_{Jss}\) to every state
                in \(B_{j}'\).
			\end{enumerate}
	\end{enumerate}
\end{definition}

A model \(\mathcal{M}'\) \(J\)-stuttering simulates a model
\(\mathcal{M}\), denoted \(\mathcal{M} \leq_{Jss} \mathcal{M}'\), if
there is a \(J\)-stuttering simulation between \(\mathcal{M}\) and
\(\mathcal{M}'\). Two models \(\mathcal{M}\) and \(\mathcal{M}'\) are
called \(J\)-stuttering simulation equivalent if \(\mathcal{M}
\leq_{Jss} \mathcal{M}'\) and \(\mathcal{M}' \leq_{Jss}
\mathcal{M}\). Any ACTL\(^*\)K\(_{-X}^{J}\) formula is preserved under
\(J\)-stuttering simulation equivalence.

\begin{theorem}
Let \(\mathcal{M}\) and \(\mathcal{M}'\) be two \(J\)-stuttering
simulation equivalent models.  Then, \((\mathcal{M},\iota) \models
\varphi \text{ iff } (\mathcal{M}',\iota') \models \varphi \), for any
ACTL\(^*\)K\(_{-X}^{J}\) formula \(\varphi\).
\end{theorem}
\begin{proof}
    Stuttering simulation equivalence is known to preserve
    ACTL\(^*_{-X}\) formulae~\cite{penczek+00}. Since atomic
    propositions in ACTL\(^{*J}_{-X}\) formulae refer only to agents
    \(J \subseteq \mathcal{A}\), \(J\)-stuttering simulation
    equivalence preserves ACTL\(^{*J}_{-X}\) formulae.  Using
    induction on the structure of \(\varphi\) it is easy to show that
    \(J\)-stuttering simulation equivalence also preserves
    ACTL\(^*\)K\(_{-X}^{J}\) formulae.
\end{proof}

It follows that if we are able to show that the cutoff instance
\(\mathcal{T}(c)\) is \(J\)-stuttering equivalent to an arbitrary
system instance \(\mathcal{T}(n)\), then we can use \(\mathcal{T}(c)\)
to check whether a formula \(\bigwedge_{J}\varphi(J) \in \)
IACTL\(^*\)K\(_{-X}\) holds for an arbitrary number of agents.

\section{MODEL CHECKING piis}

We now present a technique for model checking parameterised interleaved
interpreted systems. In particular, we propose an efficient and
automated methodology for answering the following verification query:
\[ \forall n \geq |J| : \mathcal{T}(n) \models \psi, \text{ where } \psi =
    \bigwedge_{J} \varphi(J) \in  \text{IACTL}^*\text{K}_{-X}\] 
In other
words we would like to check whether the property $\varphi(C)$ holds for \emph{any number}
\(n \geq |J|\) of agents in the system, and for any \(|J|\)-tuple \(C\) of
distinct agents. Note that the number of systems we would like to verify is unbounded.
Therefore, traditional techniques to handle the state explosion
problem cannot be used here. 
The key observation is
that in certain circumstances it is sufficient to analyse only a
finite number of systems to deduce properties about any larger system.
Inspired by the work on cutoffs in the context of reactive
systems~\cite{emerson+00,emerson+95,hanna+10}, we say that a \emph{MAS
cutoff} \(c\) is a value of the system parameter for which the system
instance \(\mathcal{T}(c)\) exhibits all the behaviour admitted by any
system instance \(\mathcal{T}(n)\), \(n \geq c\), 
with respect to a certain specification being considered.

\begin{definition}{(MAS Cutoff)}
    Let \(\mathcal{T}(n)\) be a parameterised interleaved interpreted
    system and let \(\psi \in\) IACTL\(^*\)K\(_{-X}\) be of the form
    \(\bigwedge_{J}\varphi(J)\). A natural number \(c \geq |J|\) is
    said to be a \emph{MAS cutoff} for \(\psi\) if \(\mathcal{T}(c)
    \models \psi \Leftrightarrow \forall n \geq c : \mathcal{T}(n) \models
    \psi\).
\end{definition}

It follows that if a cutoff can be identified, then model checking an
infinite family of systems can be reduced to model checking all system
instances up to the cutoff. Cutoff identification methodologies are
typically accompanied by the following shortcomings: (i) either the
cut-off is not guaranteed to be the
smallest~\cite{emerson+00,emerson+95}, or (ii) the cut-off is not
guaranteed to exist leading to incomplete
methodologies~\cite{hanna+10}.  By contrast, we here present a sound
and complete methodology for identifying the \emph{smallest cutoff} in
model checking PIIS. As we will see later, to achieve this we pay a
price in terms of the range of systems we can apply our results to.
The following lemma shows that the smallest cutoff for an
ACTL\(^*\)K\(_{-X}^{\{1,\cdots,k\}}\) formula \(\varphi\) is precisely
\(k\), the total number of agents appearing in the epistemic
modalities and the propositions. 

\begin{lemma} \label{lemma:actl*k-se}
    If \(\varphi(\{1,\cdots,k\})\) is an
    ACTL\(^*\)K\(_{-X}^{\{1,\cdots,k\}}\) formula,
    then
	\(
        \mathcal{T}(n) \models
        \varphi(\{1,\cdots,k\}) \text{ iff } \mathcal{T}(k) \models
        \varphi(\{1,\cdots,k\})
    \), for all \(n \geq k\).
\end{lemma}
\begin{proof}	

    Choose an arbitrary \(n \geq k\). Let \([n] = \{1,\cdots,n\}\) and
    \([k+1,n] = [n] \setminus [k]\).  We show that \(\mathcal{T}(n)
    \leq_{[k]ss} \mathcal{T}(k)\) and \(\mathcal{T}(k) \leq_{[k]ss}
    \mathcal{T}(n)\). The lemma then follows.

    (\(\Rightarrow (\mathcal{T}(n) \leq_{[k]ss} \mathcal{T}(k))\))
    Define a relation \(\sim_{[k]ss} = 
    \{ (g,g') \in G^{n} \times G^{k} \; | \;  g_{[k]} =
    g' \}
    \).
    We show that \(\sim_{[k]ss}\) is a \([k]\)-stuttering simulation
    between \(\mathcal{T}(n)\) and \(\mathcal{T}(k)\). Let
    \(g\sim_{[k]ss}g'\). Suppose that \(g_i = g^1_i\) for
    some \(i \in [k]\) and let \(g'^1 = g^1_{[k]}\). We have that
    \(g'_i = g'^1_i\) and \(g^{1} \sim_{[k]ss}g'^{1}\). Now
    let \(\pi=g^{1}a^{1}g^{2}a^{2}g^{3} \cdots\) be a
    path of \(\mathcal{T}(n)\) originating from \(g^1=g\).  We construct a
    path \(\rho\) of \(\mathcal{T}(k)\) originating from \(g'\)  as
    required by \([k]\)-stuttering simulation. Let \(\rho =
    g^{1}_{[k]}a'^{1}g^{2}_{[k]}a'^{2}g^{3}_{[k]} \cdots\), where
    \(a'^{j}=a^{j}\) if \(a^{j} \in \bigcup_{i \in [k]} Act_{i}\) and
    \(a'^{j}=\epsilon\) otherwise,  be the sequence obtained by the
    projection of \(\pi\) onto \([k]\). By assumption on the joint
    silent action, \(\rho\) is a valid path of \(\mathcal{T}(k)\). We
    define a partition  \(B_{1},B_{2},\cdots\) of \(\pi\) and  a
    partition \(B'_{1},B'_{2},\cdots\) of \(\rho\) such that
    \(|B_{j}|=|B'_{j}|=1\) for each \(j \geq 1\). It follows that
    \(B_{j} \sim_{[k]ss} B'_{j}\) for each \(j \geq 1\). Therefore,
    \(\mathcal{T}(n) \leq_{[k]ss} \mathcal{T}(k)\).
    
    (\(\Leftarrow (\mathcal{T}(k) \leq_{[k]ss} \mathcal{T}(n)))\) 
The idea is to allow every agent \(i \in [k+1,n]\) in \(\mathcal{T}(n)\)
to mimic agent 1 (in \(\mathcal{T}(n)\)).  For this purpose, define a
relation \(\sim_{[k]ss}\) by
    \begin{eqnarray*}
        \left\{ (g,g') \in G^{k} \times G^{n} \; | \; g =
        g'_{[k]}  \wedge \exists a^*
        \in Act^{A} :  \forall i \in [k+1,n]  :   \right. \\ \left.
        \left( a_i^* \in P_i(g'_i) \wedge
        tl(t_i(g'_i,a^*_i))=tl(g'_1) \right) \vee tl(g'_i) =
        tl(g'_1)
    \right\}
	\end{eqnarray*}

%
    If \(g \sim_{[k]ss}g'\), then each agent \(i \in [k+1,n]\) in \(g'\)
    is either at the same local state with agent \(1\) in \(g'\) or
    agent \(i\) is able to change its state to the state of agent
    \(1\) by performing the asynchronous action \(a^*_i\). We show
    that \(\sim_{[k]ss}\) is a \([k]\)-stuttering simulation between
    \(\mathcal{T}(k)\) and \(\mathcal{T}(n)\). Let
    \(g\sim_{[k]ss}g'\). Simulation requirement 2(a) follows by  a
    similar argument used in the left to right direction of the lemma.
    For simulation requirement 2(b), note that since the global
    evolution function is deterministic, a path \(g^{1}a^{1}g^{2}a^{2}
    \cdots\) is uniquely defined by the sequence
    \(g^{1}a^{1}a^{2}\cdots\).  We inductively define a function \(f\)
    which maps a path \(\rho=g^1a^1g^2a^2g^3\cdots\), \(g^1=g\), in
    \(\mathcal{T}(k)\) into a path in \(\mathcal{T}(n)\). 

    \begin{itemize}      
        \item \(f(g^1a^1g^2a^2g^3\cdots) = g'a^*_{j_1}\cdots
            a^*_{j_d}f(a^1g^2a^2g^3)\), where \(\{j_1\cdots j_d\} =
            \{i \in [k+1,n] \; | \; tl(g'_i) \neq tl(g'_1)\}\); 
        
        \item \(f(a^1g^2a^2g^3\cdots) = a^1f(a^2g^3\cdots)\), if
            \(a^1 \notin Act_1^A\);

        \item  \(f(a^1g^2a^2g^3\cdots) = a^1tl(a^1)_{k+1}\cdots
            tl(a^1)_nf(a^2g^3\cdots)\), if \(a^1 \in Act_1^A\);
   \end{itemize}
                                                                                          
    We partition \(\pi\) into singleton blocks \(B_1,B_2,\cdots\) and
    we partition \(f(\pi)=g^1a^1\cdots\) into the sequence
    \(B'_1,B'_2,\cdots\), where \(B'_i = g^j\), if \(a^{j-1} \in
    \bigcup_{z \in [2,k]} Act_z\), and \(B'_i = g^j\cdots g^{j+d}\),
    if \(a^{j-1},\cdots,a^{j+d-1} \in \bigcup_{z \in \{1\} \cup [k+1,n]}
    Act_z\) and \(a^{j+d} \in \bigcup_{z \in [k]}\) \(Act_z\). It follows
    that \(B_j \sim_{[k]ss} B'_j\), therefore, \(\mathcal{T}(k)
    \leq_{[k]ss} \mathcal{T}(n)\).	
\end{proof}

A consequence of the above lemma is the following:

\begin{theorem} \label{theorem:cut-off}
Let \(\psi\) be an  IACTL\(^*\)K\(_{-X}\) formula of the form
\(\bigwedge_{J} \varphi(J)\). Then, \( \forall n \geq |J| :
\mathcal{T}(n) \models \psi \text{ iff } \mathcal{T}(|J|) \models \psi
\).
\end{theorem}
\begin{proof}
By exploiting symmetry (Lemma~\ref{lemma:symmetry-reduction}), it
suffices to prove the result for \(\varphi([|J|])\)
(Lemma~\ref{lemma:actl*k-se}).
\end{proof}

The above theorem is our main theoretical result. 
It follows that to verify a formula \(\bigwedge_{J} \varphi(J)\) on
all system instances, it suffices to verify the formula
\(\varphi([J])\) for the system instance \(\mathcal{T}(|J|)\). Since
in the MAS literature most properties are expressed by using one or
two epistemic and propositional  indices, this dramatically improves
our verification abilities. Furthermore we can combine this technique
with others available in the literature. Specifically, upon obtaining the
instance \(\mathcal{T}(|J|)\) we can further apply partial order
reductions~\cite{lomuscio+10a}, abstraction~\cite{cohen2009abstraction}, data
symmetry reduction~\cite{cohen2009symmetry}, etc., to
further reduce the state space of the model. 


\begin{corollary}
Model checking parameterised interleaved interpreted systems against
IACTL\(^*\)K formulae of the form \(\bigwedge_J \varphi(J)\) is decidable.
\end{corollary}
\begin{proof}
    By Theorem~\ref{theorem:cut-off}, it suffices to model check the
    system instance of \(|J|\) agents against \(\varphi([|J|])\).
\end{proof}

The above is in line with literature in reactive
systems~\cite{clarke2006environment,emerson+00,emerson+95,emerson2003model}
where, although verification of parameterised systems is known to be
undecidable in general~\cite{krzysztof+86}, decidable fragments have
been obtained by imposing restrictions on the systems and the
properties studied.

\section{Example: Autonomous Robots}


For illustration purposes we exemplify the theory presented above on
a parameterised variant of the autonomous robot example. 
We assume an arbitrary number of robots each running along its own
track and each equipped with its own faulty sensor. The environment
may move all non-halted robots one position forward at each time step.
We represent this scenario by means of PIIS\@.  
We arbitrarily choose eight distinct locations; note that the number
of locations does not affect the scenario as long as it is greater
than four.  Since we use interleaving semantics, we assume that the
environment moves each robot in sequence; however, we insist on the
environment to move all robots before moving a robot twice.

We proceed to define the template agent \(\mathcal{T}\). A template
state is a 4-tuple \(l = (p,s,h,m)\), where \(p\) and \(s\) represent
the position of the robot and the value of its sensor respectively,
\(h\) represents whether or not the robot has halted, and \(m\) is a
binary variable representing whether or not the environment has moved
the robot in an interleaving sequence (a sequence in which the
environment moves all non-halted robots from position \(q\) to
\(q+1\)).  Therefore,  \( L = \left\{ (p,s,h,m) \; | \; 0 \leq p,s
    \leq 7 \text{ and } h,m \in \{\top,\bot\} \right\}\) is the set of
    template states from which we define \(\iota = (0,0,\bot,\bot)\)
    as the initial template state. 
    
A robot can either do nothing or halt; the set \(Act^{A}\) of
asynchronous template actions is \(Act^{A} = \{null^=,null^{+},\)
\(null^{-},halt\}\); the \(null\) actions represent the environment
moving the robot a position forward and either providing a correct
sensor reading (\(null^=\)) or not (\(null^{+},null^{-}\)). A robot can
move to position \(q+1\) only if all non-halted robots have moved to
position \(q\); the unique synchronous template action \( n\_s
\) (next step) synchronises all robots before the environment can move a robot.
As it will be clear below, when a \(null\) action is performed at position
\(q\), then \(m\) is set to \(\top\) and the protocol selects the
action \(n\_s\) thereby disallowing a robot to move at position
\(q+1\) before all robots have moved to position \(q\).
    
The template protocol \(P\) selects one of the \(null\) actions at
position \(q\) when \(m=\bot\) and the sensor reading is less than 3.
The synchronous action \(n\_s\) is the only allowed action when
\(m=\top\). Whenever the sensor reading is greater than 2, the halting
condition is satisfied; therefore, the protocol selects the \(halt\)
action: \(P( (p<7,s<3,h=\bot,m=\bot) ) =
\{null^=,null^{+},null^{-}\}\); \( P( (p=*,s=*,h=*,m=\top) ) =
\{n\_s\}\); \( P( (p=*,s \geq 3, h=\bot, m=\bot) ) = \{halt\}\),
where \(*\) expresses any value. The template evolution function
contains the following transitions: 
\((p,s,\bot,\bot) \stackrel{null^=}{\rightarrow} (p+1,p+1,\bot,\top)\); 
\((p,s,\bot,\bot) \stackrel{null^{+}}{\rightarrow}
(p+1,p+2,\bot,\top)\); 
\( (p,s,\bot,\bot) \stackrel{null^{-}}{\rightarrow}
(p+1,p,\bot,\top)\); 
\((p,s,\bot,\top) \stackrel{n\_s}{\rightarrow} (p,s,\bot,\bot)\) and
\((p,s,\bot,\bot) \stackrel{halt}{\rightarrow} (p,s,\top,\top)\).


We introduce the following atomic propositions: \(\textit{AP} =
\{h \text{(halted)},g \text{(goal region)},r \text{(restricted region)}\} \).  The interpretation of these propositions is
given by the following valuation function: \(V(h) = \{l \in L \;
| \; l_3=\top\}\), \(V(g) = \{l \in L \; | \; 2 \leq l_1 \leq
4\}\), and \(V(r) = \{l \in L \; | \; 0 \leq l_1 \leq 1\}\).

We verify that the halting condition is sound and complete in the
parameterised variant  by verifying  the
formulae \(\varphi_{AR2} = \bigwedge_{\{i\}} AG( g_i \rightarrow
AG(g_i))\) and \(\varphi_{AR3} = \bigwedge_{\{i\}} AG ( r_i
\rightarrow \neg h_i )\). The  specification \(\varphi_{AR2}\)
expresses that ``for every robot \(i\), if \(i\) is within the goal region,
then \(i\) never exits the goal region''.  The formula \(\varphi_{AR3}\)
states that ``for every robot \(i\), if \(i\) is within the restricted
region, then \(i\) has not halted''.  Note that the combined state
space for the systems to be checked is unbounded. Observe also that
model checking the above specifications is equivalent to model
checking the formulae \(AG(g_1 \rightarrow AG(g_1)) \wedge \cdots
\wedge\) \(AG(g_n \rightarrow AG(g_n))\) and \(AG( r_1
\rightarrow \neg h_1 ) \wedge \cdots \wedge AG(r_n \rightarrow \neg
h_n),\) on each system instance \(\mathcal{T}(n)\), \(n \geq 1\).
This is clearly not possible to check via standard model checking
techniques. However, by using Lemmas~\ref{lemma:symmetry-reduction}
and~\ref{lemma:actl*k-se} we can deduce that the MAS cutoff is equal
to~1 and reduce the problem to checking the formulas \(AG(g_1
\rightarrow AG(g_1))\) and \(AG(r_1 \rightarrow \neg h_1)\) on
the system instance \(\mathcal{T}(1)\). This is a simple problem: 
we can easily check the specifications are verified, thereby
deducing that the parameterised queries are also satisfied.


To proceed in our analysis further, we can also verify \( \varphi_{AR1} \). Also we
could check that a robot knows that every other robot halted at the
same time: \( \varphi_{AR4} = \bigwedge_{\{i,j\}} AG( h_{i}
\rightarrow K_{i}h_{j} ) \). Similarly to what above, the formulae
\(\varphi_{AR1}\) and \(\varphi_{AR4}\) can be reduced through
Lemma~\ref{lemma:symmetry-reduction} to \(\varphi'_{AR1} = AG(h_{1}
\rightarrow K_{1}g_{2})\) and \(\varphi'_{AR4} = AG(h_{1} \rightarrow
K_{1}h_{2})\), which can be verified on the system instance
\(\mathcal{T}(2)\) obtained by using a cutoff equal to~2 through
Lemma~\ref{lemma:actl*k-se}.  Also in this case we can check the
result on the much smaller model and verify that the formula
\(\varphi'_{AR1}\) holds while \(\varphi'_{AR4}\) does not~(since the
sensor readings may differ). So we infer that  \(\varphi_{AR1}\) holds
on the unbounded system while \(\varphi_{AR4}\) does not.


\section{EVALUATION}

\paragraph{Implementation}
We have implemented the presented methodology as an extension to the
open-source model checker
\texttt{\textsc{mcmas}}~\cite{lomuscio2009mcmas}. The extended model
checker, also open-sourced and named \texttt{\textsc{mcmas-p}}, is available
from~\cite{mcmasp}. ISPL, the input language of
\texttt{\textsc{mcmas}}, was suitably extended to allow for the
definition of PIIS and to support  the specification of indexed
formulae. The description of a PIIS in this language (called PISPL)
includes the declaration of a template agent. This declaration differs
from agent declarations in ISPL by having sections of asynchronous and
synchronous actions, and an initial state section. 
The specifications supported by \texttt{\textsc{mcmas-p}}
are expressed in indexed ACTLK\(_{-X}\)\@.  

Given a PIIS and a formula to be verified, \texttt{\textsc{mcmas-p}}
determines the cutoff \(c\) for the system as in
Theorem~\ref{theorem:cut-off} by counting the number of unique indices
used in the specification to be tested. A concrete system of \(c\)
agents, each an indexed copy of the template agent, is then
automatically constructed and represented symbolically.  The
specifications are automatically reduced to formulae in
ACTLK\(_{-X}\), as described in
Lemma~\ref{lemma:symmetry-reduction}. The OBDD-based algorithms
utilised by \texttt{\textsc{mcmas}} are then used to verify the system
against the reduced ACTLK\(_{-X}\) formulae.  The
BDD encoding of the joint protocol is different from that of
\texttt{\textsc{mcmas}} to enforce the interleaving semantics used
here.

\paragraph{Experimental Results} 
In order to evaluate the methodology
presented, we 
considered the parameterised autonomous robot scenario against the specification
\(\varphi_{AR1}\). We used a PC  with an Intel Core i7 processor
clocked at 2.20 GHz, with 6144 KiB cache, and running 64-bit
Fedora~17, kernel 3.3.4. The results are reported in
Table~\ref{table:exp-results}. The \emph{Robots} and \emph{States}
columns respectively show the system instance (number of robots) and
its state space; the \emph{Instantiations} column shows the number of
the possible instantiations of \(\varphi_{AR3}\), each to be verified
by the model checker; the \emph{Time} and \emph{Memory} columns show
the CPU time and memory usage respectively.  These results show that,
as expected, the state space and the length of the formulae to be
verified grow exponentially with the number of agents in the system.
As a consequence of this,
verification quickly becomes unfeasible under the time and memory
constraints. This is exemplified for the system of 90~robots, where
\texttt{\textsc{mcmas}} did not finish the model construction within the timeout of one
hour. In addition to this, of course,  plain model checking
cannot ever ensure the property holds on a system of
arbitrary many agents. In comparison \texttt{\textsc{mcmas-p}}
constructed and verified a system of 2 robots in under 0.1 seconds
thereby showing the property holds for an unbounded number of agents.

\begin{center}
    \begin{table*}
    \centering
\begin{tabular}{cccccc}
        \toprule
        \multicolumn{2}{c}{Model}  & Instantiations & Time (s) & Memory (KiB) \\
        \cmidrule(r){1-2}
        Robots & States \\
        \midrule
        2 & 201 & 2 & 0 & 9010 \\
        30 & \num{1.26092e27} & 870 & 18 & 44744 \\
        60 & \num{3.59402e57} & 3540 & 868 & 63894032 \\
        90 & \texttt{\textsc{timeout}} & 8010 & \texttt{\textsc{timeout}} &
        \texttt{\textsc{timeout}} \\ 
        \bottomrule
    \end{tabular}
\caption{\texttt{\textsc{mcmas}} verification results for \(\varphi_{AR1}\).}
\label{table:exp-results}

\vspace{-0.5cm}
\end{table*}
\end{center}

\section{CONCLUSIONS AND RELATED WORK}

In this paper we have developed a technique to verify that a
temporal-epistemic property holds in a MAS irrespective of the
number of agents present in the system. The problem is undecidable in
general but we have defined a suitable semantics  for
which we gave a sound and complete procedure for determining a cutoff
for a system. To do so we have defined a suitable parameterised logic
and developed stuttering-equivalence simulation results for it on
PIIS.  We find the result encouraging as it opens the way for the
verification of a large number of protocols previously verified only
for individual instances containing a limited number of agents. Open
systems with an unbounded number of homogeneous participants, e.g.,
including negotiations and auctions, seem particularly suitable for
this analysis.

\paragraph{Related Work} 
Existing literature on parameterised
verification~\cite{clarke2006environment,pnueli2002liveness,wolper1990verifying,emerson+95,emerson+00,hanna+10}
is limited to reactive systems and plain temporal logics.  Moreover,
mainstream methodologies~\cite{emerson+95,emerson+00,hanna+10} do not
guarantee soundness, completeness and the identification of the
smallest cutoff at the same time.  In~\cite{hanna+10} a cutoff is
identified by enumerating the system instances and finding the
smallest instance able to simulate a ``special'' structure which
includes the behaviour of every instance.  Although the technique is
widely applicable and independent of the communication topology, a
cutoff is not guaranteed to exist. Results closer to those in this
paper are the sound and complete techniques put forward
in~\cite{emerson+95,emerson+00}. Similarly to this contribution,
\cite{emerson+95,emerson+00} present stuttering-simulation results
between the cutoff model and every system instance thereby ensuring
soundness and completeness; however, the results in~\cite{emerson+95}
are applicable to ring topologies only and the technique
in~\cite{emerson+00} does not identify the smallest cutoff.

In addition to cutoffs, abstraction techniques have of course  been used in
parameterised verification. In~\cite{pnueli2002liveness}
concrete states are counter abstracted; an abstract state is a tuple
of counters, one for each local state, denoting the number of system
participants in the state. This process can be automated, but it is
only applicable to a narrow class of systems and it is restricted to
liveness properties. Environmental
abstraction~\cite{clarke2006environment} extends counter abstraction
by counting the number of participants that satisfy a given predicate
and, although achieving wider applicability, the methodology has not,
to our knowledge, 
been automated yet. In~\cite{wolper1990verifying} a \emph{network
  invariant} is identified which exhibits the behaviour of all system
instances; if the invariant satisfies a property, then the property is
satisfied by all system instances. A network invariant, however, is
not guaranteed to exist, and, even when it does, its identification is
not automated. In addition, none of these works tackle
epistemic logic, nor MAS semantics, as we do here.

\paragraph{Future Work}
A current limitation of the PIIS formalism is that agents cannot
evolve differently depending on the environment's action. This limits
the application of the technique to particular systems such different
network topologies. In future work we plan to alleviate this
limitation as well as apply the methodology here presented to
protocols of practical interests.




\section{Acknowledgments}
The research described in this paper was  supported by the EPSRC Research
Project ``Trusted Autonomous  Systems'' (grant No. EP/I00520X/1). 
\addcontentsline{toc}{chapter}{Bibliography}
\bibliographystyle{plain}
\bibliography{bib.bib}

\begin{thebibliography}{10}

\bibitem{mcmasp}
{MCMAS-P}, {M}odel {C}hecking {P}arameterised {M}ulti-{A}gent {S}ystems.
  \url{http://vas.doc.ic.ac.uk/software/tools/}.

\bibitem{krzysztof+86}
K.R. Apt and D.C. Kozen.
\newblock Limits for automatic verification of finite-state concurrent systems.
\newblock {\em Information Processing Letters}, 22:307--309, 1986.

\bibitem{clarke2006environment}
E.~Clarke, M.~Talupur, and H.~Veith.
\newblock Proving ptolemy right: The environment abstraction framework for
  model checking concurrent systems.
\newblock In {\em Proc. of TACAS'08}, pages 33--47. Springer, 2008.

\bibitem{clarke+89a}
E.M. Clarke, O.~Grumberg, and M.C. Browne.
\newblock Reasoning about networks with many identical finite state processes.
\newblock {\em Inf. and Comput.}, 81(1):13--31, 1989.

\bibitem{clarke+99a}
G.~Clarke, O.~Grumberg, and D.A. Peled.
\newblock {\em Model Checking}.
\newblock The MIT Press, 1999.

\bibitem{cohen2009symmetry}
M.~Cohen, M.~Dam, A.~Lomuscio, and H.~Qu.
\newblock A symmetry reduction technique for model checking temporal-epistemic
  logic.
\newblock In {\em Proc. of IJCAI'09}, pages 721--726, 2009.

\bibitem{cohen2009abstraction}
M.~Cohen, M.~Dam, A.~Lomuscio, and F.~Russo.
\newblock Abstraction in model checking multi-agent systems.
\newblock In {\em Proc. of AAMAS'09}, pages 945--952, 2009.

\bibitem{emerson+00}
E.~Emerson and V.~Kahlon.
\newblock Reducing model checking of the many to the few.
\newblock In {\em Proc. of CADE'00}, pages 236--254. Springer, 2000.

\bibitem{emerson2003model}
E.A. Emerson and V.~Kahlon.
\newblock Model checking guarded protocols.
\newblock In {\em Proc. of LICS'03}, pages 361--370. IEEE, 2003.

\bibitem{emerson+95}
E.A. Emerson and K.S. Namjoshi.
\newblock Reasoning about rings.
\newblock In {\em Proc. of POPL'95}, pages 85--94. ACM, 1995.

\bibitem{emerson+96b}
E.A. Emerson and A.P. Sistla.
\newblock Symmetry and model checking.
\newblock {\em Formal methods in system design}, 9(1):105--131, 1996.

\bibitem{fagin+03a}
R.~Fagin, Y.~Moses, J.Y. Halpern, and M.Y. Vardi.
\newblock {\em Reasoning about knowledge}.
\newblock The MIT Press, 2003.

\bibitem{gammie2004mck}
P.~Gammie and R.~Van Der~Meyden.
\newblock Mck: Model checking the logic of knowledge.
\newblock In {\em Proc. of CAV'04}, pages 256--259. Springer, 2004.

\bibitem{hanna+10}
Y.~Hanna, D.~Samuelson, S.~Basu, and H.~Rajan.
\newblock Automating cut-off for multi-parameterized systems.
\newblock In {\em Proc. of ICFEM'10}, pages 338--354. Springer, 2010.

\bibitem{kacprzak2008verics}
M.~Kacprzak, W.~Nabia{\l}ek, A.~Niewiadomski, et~al.
\newblock Verics 2007-a model checker for knowledge and real-time.
\newblock {\em Fundam. Inform.}, 85(1):313--328, 2008.

\bibitem{lomuscio+10a}
A.~Lomuscio, W.~Penczek, and H.~Qu.
\newblock Partial order reductions for model checking temporal-epistemic logics
  over interleaved multi-agent systems.
\newblock {\em Fundam. Inform.}, 101(1):71--90, 2010.

\bibitem{lomuscio2009mcmas}
A.~Lomuscio, H.~Qu, and F.~Raimondi.
\newblock Mcmas: A model checker for the verification of multi-agent systems.
\newblock In {\em Proc. of CAV'09}, pages 682--688. Springer, 2009.

\bibitem{penczek+00}
W.~Penczek, M.~Szreter, R.~Gerth, and R.~Kuiper.
\newblock Improving partial order reductions for universal branching time
  properties.
\newblock {\em Fundam. Inform.}, 43(1):245--267, 2000.

\bibitem{pnueli2002liveness}
A.~Pnueli, J.~Xu, and L.~Zuck.
\newblock Liveness with (0, 1,infinity)-counter abstraction.
\newblock In {\em Proc. of CAV'02}, pages 93--111. Springer, 2002.

\bibitem{wolper1990verifying}
P.~Wolper and V.~Lovinfosse.
\newblock Verifying properties of large sets of processes with network
  invariants.
\newblock In {\em Automatic Verification Methods for Finite State Systems},
  pages 68--80. Springer, 1990.

\end{thebibliography}

\end{document}